\documentclass[a4paper]{article}

\usepackage[utf8]{inputenc}

\usepackage{fullpage}
\usepackage{amsfonts,amsmath,amssymb,amsthm}
\usepackage[main=american]{babel}
\usepackage{xspace}

\usepackage[inline]{enumitem}

\usepackage{graphicx}

\usepackage[hidelinks,pdfborder={0 0 0}]{hyperref} 
\usepackage{cleveref}

\usepackage{cleveref}
\usepackage[boxruled,linesnumbered]{algorithm2e}

\newtheorem{theorem}{Theorem}[section]

\newtheorem{lemma}[theorem]{Lemma}
\newtheorem{corollary}[theorem]{Corollary}

\newtheorem{definition}[theorem]{Definition}
\newenvironment{acks}{\paragraph{Acknowledgments:}}{}

\newcommand{\tset}{{\mathcal T}}
\newcommand{\poly}{\operatorname{poly}}

\newcommand{\scp}{\mbox{\sf Sparsest-Cut}\xspace}
\newcommand{\chlamtac}{Chlamt{\'a}{\v{c}}\xspace}
\newcommand{\chlamtacetal}{\chlamtac et al.\@\xspace}
\newcommand{\cdiam}{\ensuremath{\delta}\xspace}
\newcommand{\lpcut}{\operatorname{\sf lpcut}\xspace}
\makeatletter
\newcommand{\DeclarePairedDelimiter}[3]{
    \DeclareRobustCommand*{#1}[2][\auto]{%
        \ifx\auto##1\mathopen{}\left#2##2\right#3\mathclose{}\else##1#2##2##1#3\fi}%
    }%
\makeatother
\DeclarePairedDelimiter\set{\lbrace}{\rbrace}
\DeclarePairedDelimiter\paren{(}{)}
\newcommand{\aset}{{\mathcal A}}
\newcommand{\algcut}{\operatorname{\sf algcut}\xspace}
\newcommand{\diam}{\ensuremath{d}\xspace}
\newtheorem{observation}[theorem]{Observation}

\DeclareMathOperator{\cp}{cap}
\DeclareMathOperator{\dm}{dem}

\begin{document}

\title{Approximating Sparsest Cut in Low-Treewidth Graphs via Combinatorial Diameter}

\author{
  Parinya Chalermsook \thanks{Aalto University, Finland. {\bf email:} \texttt{parinya.chalermsook@aalto.fi}} \and
  Matthias Kaul \thanks{{Technische Universität Hamburg}, {Germany}. \textbf{email:} \texttt{matthias.kaul@tuhh.de}} \and
  Matthias Mnich \thanks{{Technische Universität Hamburg}, {Germany}. \textbf{email:} \texttt{matthias.mnich@tuhh.de}} \and
  Joachim Spoerhase \thanks{{Aalto University}, {Finland}. \textbf{email:} \texttt{joachim.spoerhase@aalto.fi}} \and
  Sumedha Uniyal \thanks{{Aalto University}, {Finland}.} \and
  Daniel Vaz \thanks{{Operations Research, Technische Universität M{\"u}nchen}, {Germany}. \textbf{email:} \texttt{daniel.vaz@tum.de}}
}

\maketitle

\begin{abstract}
  The fundamental sparsest cut problem takes as input a graph $G$ together with the edge costs and demands, and seeks a cut that minimizes the ratio between the costs and demands across the cuts.
  For $n$-node graphs~$G$ of treewidth~$k$, \chlamtac, Krauthgamer, and Raghavendra (APPROX 2010) presented an algorithm that yields a factor-$2^{2^k}$ approximation in time $2^{O(k)} \cdot \poly(n)$. 
  Later, Gupta, Talwar and Witmer (STOC 2013) showed how to obtain a $2$-approximation algorithm with a blown-up run time of $n^{O(k)}$. 
  An intriguing open question is whether one can simultaneously achieve the best out of the aforementioned results, that is, a factor-$2$ approximation in time $2^{O(k)} \cdot \poly(n)$.

   In this paper, we make significant progress towards this goal, via the following results:
   \begin{itemize}
     \item[(i)] A factor-$O(k^2)$ approximation that runs in time $2^{O(k)} \cdot \poly(n)$, directly improving the work of Chlamt{\'a}{\v{c}} et al.\ while keeping the run time single-exponential in $k$.
     \item[(ii)] For any $\varepsilon>0$, a factor-$O(1/\varepsilon^2)$ approximation whose run time is $2^{O(k^{1+\varepsilon}/\varepsilon)} \cdot \poly(n)$, implying a constant-factor approximation whose run time is nearly single-exponential in $k$ and a factor-$O(\log^2 k)$ approximation in time $k^{O(k)} \cdot \poly(n)$.
  \end{itemize}
  Key to these results is a new measure of a tree decomposition that we call \emph{combinatorial diameter}, which may be of independent interest. 
\end{abstract}

\section{Introduction}
\label{sec:intro}
In the sparsest cut problem, we are given a graph together with costs and demands on the edges, and our goal is to find a cut that minimizes the ratio between the costs and demands across the cut. 
Sparsest cut is among the most fundamental optimization problems that has attracted interests from both computer scientists and mathematicians. Since the problem is $\mathsf{NP}$-hard~\cite{MatulaS90}, the focus has been to study approximation algorithms for the problem.
Over the past four decades, several breakthrough results have eventually culminated in a factor-$\tilde O(\sqrt{\log n})$ approximation in polynomial time~\cite{arora2008euclidean,arora2009expander,leighton1999multicommodity}.
On the lower bound side, the problem is $\mathsf{APX}$-hard~\cite{chuzhoy2009polynomial} and, assuming the Unique Games Conjecture, does not admit any constant-factor approximation in polynomial time~\cite{chawla2006hardness}. 

The extensive interest  in  sparsest cuts stems from both applications and mathematical reasons.
From the point of view of applications, the question of partitioning the universe into two parts while minimizing the ``loss'' across the interface is crucial in any divide-and-conquer approach e.g., in image segmentation.
From a mathematical/geometric viewpoint, the integrality gap of convex relaxations for sparsest cuts is equivalent to the embeddability of  any finite metric space  (for LP relaxation) and  of any negative-type metric (for SDP relaxation)\footnote{A metric $(X,d)$ is said to be \textit{negative type}, if $(X,\sqrt{d})$ embeds isometrically into a Hilbert space.} into $\ell_1$.
Therefore, it is not a surprise that this problem has attracted interest from both computer science and mathematics (geometry, combinatorics, and functional analysis) communities.

The study of sparsest cuts in the low-treewidth regime was initiated in 2010 by \chlamtac, Krauthgamer, and Raghavendra~\cite{ckr10}, who devised a factor-$2^{2^k}$ approximation algorithm (CKR) that runs in time $2^{O(k)} \cdot \poly(n)$, with $k$ being the treewidth of the input graph.
Later, Gupta, Talwar and Witmer~\cite{gtw13} showed how to obtain a factor-$2$ approximation (GTW) with a blown-up run time of $n^{O(k)}$; they further showed that there is no $(2-\varepsilon)$-approximation for any $\varepsilon > 0$ on constant-treewidth graphs, assuming the Unique Games Conjecture.
It remains an intriguing open question whether one can simultaneously achieve the best run time and approximation factor.
In particular, in this paper we address the following question: 
\begin{quote}
  Does \scp admit a factor-$2$ approximation that runs in time $2^{O(k)} \cdot \poly (n)$? 
\end{quote}

\paragraph{Broader perspectives.}
Given the significance of sparsest cuts, a lot of effort have been invested into understanding when sparsest cut instances are ``easy''. 
In trees, optimal sparsest cuts can be found in polynomial time~(see e.g.~\cite{GuptaL19}).
For many other well-known graph classes, finding optimal sparsest cuts is $\mathsf{NP}$-hard, so researchers attempted to find constant-factor approximations in polynomial time.
They have succeeded, over the past two decades, for several classes of graphs, such as outerplanar, $\ell$-outerplanar, bounded-pathwidth and bounded-treewidth graphs~\cite{gtw13,gupta2004cuts,chekuri2006embedding,ckr10,lee2013pathwidth}, as well as planar graphs~\cite{cohen2021quasipolynomial}. 

As mentioned earlier, sparsest cuts are not only interesting from the perspective of algorithm design, but also from the perspectives of geometry, probability theory and convex optimization.  
Indeed, the famous conjecture of Gupta, Newman, Rabinovich, and Sinclair~\cite{gupta2004cuts} postulates that any minor-free graph metric embeds into $\ell_1$ with a constant distortion, which would imply that all such graphs admit a constant approximation for the sparsest cut problem.
The conjecture has been verified in various graph classes~\cite{lee2013pathwidth,chekuri2006embedding}, but remains open even for bounded-treewidth graph families. 

To us, perhaps the most interesting aspect of the treewidth parameter~\cite{ckr10,gtw13} is its connection to  the power of hierarchies of increasingly tight convex relaxations (see, for instance, the work by Laurent~\cite{laurent03}).
In this setting, a straightforward (problem-independent!)  LP rounding algorithm performs surprisingly well for many ``combinatorial optimization'' problems.  
It has been shown to achieve optimal solutions for various fundamental problems in bounded treewidth graphs~\cite{magen2009robust,bienstock2004tree,wainwright2004treewidth} and match the (tight) approximation factors achievable on trees for problems such as {group Steiner tree}~\cite{chalermsook2017beyond,chalermsook2018survivable,garg2000polylogarithmic,halperin2003polylogarithmic}.   
In this way, for these aforementioned problems, such a problem-oblivious LP rounding algorithm provides a natural framework to generalize an optimal algorithm on trees to nearly-optimal ones on low (perhaps super-constant) treewidth graphs. 
Our work can be seen as trying to develop such  understanding in the context of the sparsest cut problem.

\subsection{Our Results}
We present several results that may be seen as an intermediate step towards the optimal result. 
Our main technical results are summarized in the following theorem. 

\begin{theorem}
\label{thm:main-intro} 
  For the following functions $t$ and $\alpha$, there are algorithms that run in time $t(k)\cdot \poly(n)$ and achieve approximation factors $\alpha(k)$ for the sparsest cut problem: 
  \begin{itemize}
    \item $t(k) = 2^{O(k)}$ and $\alpha(k) = O(k^2)$. 
    \item $t(k) = 2^{O(k^2)}$ and $\alpha(k) = O(1)$.       
    \item For any $\varepsilon > 0$, $t(k) = \exp\paren{O(\frac{k^{1+\varepsilon}}{\varepsilon})}$ and $\alpha(k) = O(1/\varepsilon^2)$. 
  \end{itemize}
\end{theorem}

Our first result directly improves the approximation factor of $2^{2^k}$ by \chlamtacetal, while keeping the run time single-exponential in $k$.
Our second result shows that, with only slightly more exponential run time, one can achieve a constant approximation factor.
Compared to Gupta et al., our result has a constant blowup in the approximation factor (but independent of $k$), but has a much better run time ($2^{O(k^2)}$ instead of $n^{O(k)}$); compared to \chlamtacetal, our result has a much better approximation factor ($O(1)$ instead of $2^{2^k}$), while maintaining nearly the same asymptotic run time.

Finally, our third result gives us an ``approximation scheme'' whose run time exponent converges to a single exponential, while keeping an approximation factor a constant. 
We remark that, by plugging in $\varepsilon = \Omega(1/\log k)$, we obtain a factor-$O(\log^2 k)$ approximation in time $k^{O(k)}\cdot\poly(n)$.

\subsection{Overview of Techniques} 
Now, we sketch the main ideas used in deriving our results.
We assume certain familiarity with the notions of treewidth and tree decomposition. 
Let $G$ be a graph with treewidth $k$ and $\tset$ be a tree decomposition of $G$ with a collection of bags $B_t \subseteq V(G)$ for all $t \in V(G)$.
Define the width of $\tset$ as $w(\tset) = \max_{t \in V(\tset)} |B_t| -1$. 

The run time of algorithms that deal with the treewidth parameter generally depend on $w(\tset)$, so when designing an algorithm in low-treewidth graphs, one usually starts with a near-optimal tree decomposition in the sense that $w(\tset) = O(k)$. 
To give a concrete example, the CKR algorithm~\cite{ckr10} for sparsest cut runs in time $2^{O(w(\tset))} \cdot \poly(n)$ and gives approximation factor $2^{2^{w(\tset)}}$ . 
Observe that, with slightly higher width $w(\tset) = O(\log n + \beta(k))$, the CKR algorithm would run in time $2^{\beta(k)} \cdot \poly (n)$. 

Our results are obtained via the concept of {\bf combinatorial diameter} of a tree decomposition.
Informally, the combinatorial length between $u$ and $v$ in $\tset$ measures the number of ``non-redundant bags'' that lie on the unique path in $\tset$ connecting the bags of $u$ and $v$. 
We say that the combinatorial diameter $\Delta(\tset)$ of $\tset$ is at most $d$ if the combinatorial length of every pair of vertices is at most $d$. Please refer to \Cref{sec:combinatorial diameter} for formal definitions.  

Our first key technical observation shows that the approximation factor of the CKR algorithm can be upper bounded in terms of the combinatorial diameter $\min \{ O(\Delta(\tset)^2), 2^{2^{w(\tset)}}\}$. Moreover, in the special case of $\Delta(\tset) = 1$, the CKR algorithm gives a $2$-approximation, which can be seen by using the arguments of Gupta et al.~\cite{gtw13}. 
Therefore, to obtain a fast algorithm with a good approximation factor, it suffices to prove the existence of a tree decomposition with simultaneously low $w(\tset)$ and low $\Delta(\tset)$. 
We remark that standard tree decomposition algorithms~\cite{bodlaender2016c} give us $w(\tset) = O(k)$ and $\Delta(\tset) = O(\log n)$, so this observation alone does not immediately lead to improved algorithmic results. 
However, it allows us to view the results from CKR~\cite{ckr10} and GTW~\cite{gtw13} in the same context: CKR applies the algorithm to the tree decomposition $\tset_{CKR}$ with $\Delta(\tset_{CKR}) = O(\log n)$ and $w(\tset_{CKR}) = O(k)$, while GTW applies the same algorithm with $\Delta(\tset_{GTW}) = 1$ and $w(\tset_{GTW}) = O(k \log n)$. 
In other words, the same algorithm is applied to two different ways of decomposing the input graph~$G$ into a tree.

In this paper, we present several new tree decomposition algorithms that optimize the tradeoff between $w(\tset)$ and $\Delta(\tset)$. 
Our first algorithm gives a tree decomposition $\tset_1$ with $w(\tset_1) = {O(\log n+k)}$ and $\Delta(\tset_1) = O(k)$, which leads to a factor-$O(k^2)$ approximation in time $2^{O(k)} \cdot \poly(n)$; this directly improves the approximation factor of CKR while maintaining the same asymptotic run time. 
Our second algorithm gives the tree $\tset_2$ with $w(\tset_2) = O(\log n + k^2)$ and $\Delta(\tset_2) = 4$.
This leads to an algorithm for sparsest cut with run time $2^{O(k^2)} \cdot \poly(n)$ and approximation factor $O(1)$. 
Our third algorithm is an approximation scheme which is further parameterized by $\varepsilon >0$.
In particular, for any $\varepsilon>0$, we construct the tree $\tset_{3,\varepsilon}$ such that  $w(\tset_{3,\varepsilon}) = O(\log n+ k^{1+\varepsilon}/\varepsilon)$ and $\Delta(\tset_{3,\varepsilon}) = O(1/\varepsilon)$. 

\subsection{Conclusion \& Open Problems}
Our work is an attempt to simultaneously obtain the best run time and approximation factor for sparsest cut in the low-treewidth regime.
Our research question combines the flavors of two very active research areas, namely parameterized complexity and approximation algorithms. 
We introduce a new measure of tree decomposition called combinatorial diameter and show various constructions with different tradeoffs between $w(\tset)$ and $\Delta(\tset)$. 
We leave the question of getting $2$-approximation in $2^{O(k)} \cdot \poly(n)$ time as the main open problem.
One way to design such an algorithm is to show an existence of a tree decomposition with $w(\tset) = O(\log n +k)$ and $\Delta(\tset) = 2$.
An interesting intermediate step would be to show $w(\tset) = O(\log n +f(k))$ for some function $f$ and $\Delta(\tset) = 2$, which would imply a fixed-parameter algorithm that yields a $2$-approximation.

Another interesting question is to focus on polynomial-time algorithms and optimize the approximation factor with respect to treewidth.
In particular, is there an $O(\log^{O(1)} k)$ approximation in polynomial time?
This question is open even for the \emph{uniform} sparsest cut problem (unit demand for every vertex pair), for which a fixed-parameter algorithm~\cite{BonsmaBPP2012} but no polynomial-time algorithm is known.

A broader direction that would perhaps complement the study along these lines is to improve our understanding on  a natural LP-rounding algorithm on the lift-and-project convex programs in general.
For instance, can we prove a similar tradeoff result for other combinatorial optimization problems in this setting?
One candidate problem is the {\em group Steiner tree} problem, for which a factor-$O(\log^2 n)$ approximation in time $n^{O(k)}$ is known (and the algorithm there is ``the same'' algorithm as used for finding sparsest cuts).
Can we get a factor-$O(\log^2 n)$ approximation in time $2^{O(k)} \cdot \poly(n)$?

\paragraph{Independent Work:} Independent of our work, Cohen-Addad, M\"{o}mke, and Verdugo~\cite{tobias} obtained a $2$-approximation algorithm for sparsest cut in treewidth $k$ graph with running time $2^{2^{O(k)}} \cdot \text{poly}(n)$.
Observe that their result is incomparable with our result: they obtain a better approximation factor, whereas the obtained running time is considerably larger than ours.
Similar to our result, they build on the techniques from \cite{ckr10,gtw13}.

\section{Preliminaries}
\paragraph{Problem Definition} In the \scp problem (with general demands), the input is a graph $G=(V, E_G)$ with positive edge capacities $\left \{\cp_e \right \}_{e \in E_G}$ and a demand graph $D =(V, E_D)$ (on the same set of vertices) with positive demand values $\left \{\dm_e \right \}_{e \in E_D}$. The aim is to determine
\[
\Phi_{G, D} := \min_{S \subseteq V} \Phi_{G,D}(S), \quad\quad \Phi_{G,D}(S) := \frac{\sum_{e \in E_G(S, V-S)}\cp_e}{\sum_{e \in E_D(S, V-S)} \dm_e}.\]
The value $\Phi_{G,D}(S)$ is called the \emph{sparsity} of the cut $S$.

\paragraph{Tree decomposition}
Let $G=(V, E)$ be a graph. A tree decomposition $(\tset, \set{B_t}_{t \in V(\tset)})$ of $G$ is a tree $\tset$ together with a collection of \emph{bags} $\{B_t\}_{t \in V(\tset)}$, where the bags $B_t \subseteq V(G)$ satisfy the following properties:
\begin{itemize}
	\item $V(G) = \bigcup_t B_t$.
	\item For any edge $uv \in E(G)$, there is a bag $B_t$ containing both $u$ and $v$.
	\item For each vertex $v \in V(G)$, the collection of bags that contain $v$ induces a connected subgraph of $\tset$.
\end{itemize}
The treewidth of graph $G$ is defined as the minimum integer $k$ such that there exists a tree decomposition where each bag contains at most $k+1$ vertices.

We generally use $r$ to denote the root of $\tset$, and $p\colon V(\tset) \to V(\tset)$ for the parent of a node with respect to root $r$.  %
We sometimes refer to $B_{p(i)}$ as the \emph{parent bag} of $B_i$. %
We denote by $\tset_{i \leftrightarrow j}$
the set of nodes on the unique path in tree $\tset$ between nodes $i, j \in
V(\tset)$ (possibly $i=j$). %
For a set $X \subseteq V(\tset)$ of bags, we use the shorthand $B(X) = \bigcup_{i \in X} B_i$ (the union of bags for nodes in $X$).

We will treat cuts in a graph as assignments of $\{0,1\}$ to each vertex, and fix some corresponding notation.
\begin{definition}
	Let $X$ be some finite set.
	An \emph{$X$-assignment} is a map $f\colon X \to \{0,1\}$.
	We denote by~$\mathcal{F}[X]$ the set of all $X$-assignments.
	For some distribution $\mu$ over $\mathcal{F}[X]$ and set $Y \subseteq X$ we define~$\mu|_Y$ to be the distribution given by
	\[
	\Pr_{f \sim \mu|_Y}[f = f'] = \Pr_{f\sim \mu}[f|_Y = f'] \quad \forall f' \in \mathcal{F}[Y] \enspace .
	\]
\end{definition}

\section{Algorithm and Combinatorial Diameter}
Our approach is based on the new relation between the algorithm of \chlamtacetal~\cite{ckr10} and our novel notion of ``combinatorial diameter''.
In \Cref{sec:combinatorial diameter}, we present the definition of the combinatorial diameter.
The subsequent sections give the description of \chlamtacetal and prove the relation to the combinatorial diameter.  

\subsection{Our New Concept: Combinatorial Diameter} 
\label{sec:combinatorial diameter} 
\begin{definition}[Redundant bags]
  Fix $s,t \in V(\tset)$. 
  Let $v \in V(\tset)$ be a bag with exactly two neighbors $u$ and $w$ on the path $\tset_{s \leftrightarrow t}$. When $B_v \cap B_w \subseteq B_u$, we say that $v$ is \emph{$(s,t)$-redundant}.
\end{definition}

Intuitively, each node $v$ discarded in the fashion above can be thought of as a subset of $u$, since the vertices $B_v \setminus B_u$ occur only in $B_v$ within $\tset_{s \leftrightarrow t}$. %
As a consequence, we can show that they do not affect the rounding behaviour of the CKR algorithm with respect to $s$ and $t$ (therefore ``redundant''). 

\begin{definition}[Simplification]
\label{def:simplification}
	Let $\tset$ be a tree decomposition, and $s,t \in V(\tset)$. We say $\tset_{s \leftrightarrow t}$ has \emph{combinatorial length} at most $\ell$ if it can be reduced to a path of length at most $\ell$ by repeatedly applying the following rule:
	\begin{itemize}
    \item[] Delete an $(s,t)$-redundant node $v$ on path $\tset_{s \leftrightarrow t}$, and add the edge~$\set{u,w}$. %
      We call this operation \emph{bypassing} $v$.
	\end{itemize}
	We call any path $P$ generated from $\tset_{s\leftrightarrow t}$ in this fashion a \emph{simplification} of $\tset_{s \leftrightarrow t}$.
\end{definition}

\begin{definition}[Combinatorial diameter]
  The \emph{combinatorial diameter} of $\tset$ is defined to be the minimum $\cdiam$ such that, for all $u,v$, the path $\tset_{u \leftrightarrow v}$ has combinatorial length at most $\cdiam$.    
\end{definition}

\subsection{Algorithm Description and Overview}
For completeness, we restate the essential aspects of the algorithm by \chlamtacetal~\cite{ckr10}.
The algorithm is initially provided a \scp instance $(G,D,\cp,\dm)$ alongside a  tree decomposition $\tset$ of $G$ with the width $w(\tset) = \max_{t} |B_t| - 1$.
The goal is then to compute a cut in $G$ that has low sparsity.

The algorithm starts by computing, for every vertex set $L = B_i \cup \set{s,t}$, consisting of a bag $B_i$ and a pair of vertices $s,t \in V(G)$, a distribution $\mu_L$ over $L$-assignments. %
This collection of distributions for all sets $L$ satisfies the requirement that any two distributions agree on their joint domains, i.e.~$\mu_L|_{L \cap L'} = \mu_{L'}|_{L\cap L'}$ for each pair of sets $L, L'$ with the structure above. %

 If we denote $\lpcut(s,t) = \Pr_{f \sim \mu_{B \cup \{s,t\}}}[f(s) \neq f(t)]$ for any $s,t \in V(G)$, and an arbitrary bag $B$ of~$\tset$, we can compute the collection of distributions that minimizes 
\[
	\dfrac{\sum_{\{s,t\} \in E_G}\cp_{\{s,t\}} \cdot \lpcut(s,t)   }{\sum_{\{s,t\} \in E_D}\dm_{\{s,t\}} \cdot \lpcut(s,t)   } \enspace . 
\]
Notice that $\lpcut$ is well-defined by the consistency requirement, since the choice of $B$ does not impact the distribution over $\{s,t\}$-assignments.
For ease of notation, we will refer to the implied distribution over some vertex set $X \subseteq B \cup \{s,t\}$ by $\mu_X$, where formally $\mu_X = \mu_{B \cup \{s,t\}}|_{X}$.

Such a collection of distributions can be computed in time $2^{O(w(\tset))}\poly(n)$, using Sherali-Adams LP hierarchies, which motivates the function name $\lpcut$.
It is then rounded to some $V(G)$-assignment $f$ using \autoref{alg:Chlamtac}.
We now recall a number of useful results about the algorithm and the assignment it computes.
Details about the algorithm and the attendant lemmas can be found in the work of \chlamtacetal~\cite{ckr10}.

Denote by $\aset$ the distribution over $V(G)$-assignments produced by the algorithm.
\begin{algorithm}[t]
	\SetAlgoLined
	\KwData{$G, (\tset, \set{B_i}_{i \in V(\tset)}), \{\mu_L\}$}

	Start at any bag $B_0$, sample $f|_{B_0}$ from $\mu_{B_0}$\;
	We process the bags in non-decreasing order of distance from $B_0$ \;
	\ForEach{Bag $B$ with a processed parent bag $B'$}{
		Let $B^+ = B\cap B'$ the subset of $B$ on which $f$ is fixed.
		Let $B^- := B\setminus B^+$.
		Sample $f|_{B^-}$ according to
		\[
			\Pr[f|_{B^-} = f'] = \Pr_{f^* \sim \mu_B}[f^*|_{B^-} = f'\; \mid \; f^*|_{B^+} =f|_{B^+} ] \quad\forall f' \in \mathcal{F}[B^-]
		\]

	}

	\KwResult{$f$}
	\caption{Algorithm \textsc{SC-Round} }
	\label{alg:Chlamtac}
\end{algorithm}

\begin{lemma}[\cite{ckr10}, Lemma 3.3]
	\label{lem:bagRealisation}
	For every bag $B$ the assignment $f|_B$ computed by \autoref{alg:Chlamtac} is distributed according to $\mu_B$, meaning $ \Pr_{ f \sim \aset}[f|_B = f'] = \Pr_{f^* \sim \mu_B}[f^* = f']$ for all $f' \in \mathcal{F}[B]$.
\end{lemma}

A direct consequence of this lemma is the fact that any edge $\{s,t\}$ of $G$ is cut by the algorithm with probability $\lpcut(s,t)$.
In particular, the expected capacity of the rounded cut is therefore
\[
\sum_{\{s,t\} \in E_G}\cp_{\{s,t\}} \cdot \lpcut(s,t),
\]
which is the value ``predicted'' by the distribution $\mu_L$.
The same property does not hold for the (demand) edges of $D$ since they may not be contained in any bag of $\tset$.

Denote by $\algcut(s,t)$ the probability that the algorithm separates $s$ and $t$, that is, $\algcut(s,t) = \Pr_{f \sim \aset}[f(s) \neq f(t)]$.
We would like to lower bound $\algcut(s,t) \geq c \; \lpcut(s,t)$ for all demand edges $\{s,t\}$ and some value $c > 0$.
This would imply that the expected demand of the rounded cut is at least $c\sum_{\{s,t\} \in E_D}\dm_{\{s,t\}} \lpcut(s,t)$, and having a good expected demand and capacity is sufficient for computing a good solution by the following observation.
\begin{observation}[\cite{ckr10}, Remark 4.3]
\label{obs:derandomisation}
	The cut sparsity $\alpha$ predicted by distributions $\set{\mu_L}_L$ is
	\[
	\alpha := \dfrac{\sum_{\{s,t\} \in E_G}\cp_{\{s,t\}} \cdot \lpcut(s,t)   }{\sum_{\{s,t\} \in E_D}\dm_{\{s,t\}} \cdot  \lpcut(s,t)  } \enspace .
	\]
	Then if $\algcut(s,t) \geq c \cdot \lpcut(s,t)$ for all $\{s,t\} \in E_D$ and $\algcut(s,t) = \lpcut(s,t)$ for $\{s,t\} \in E_G$, we have
	\[
		\mathbb{E}_{f \sim \aset}\left[\sum_{\{s,t\} \in E_G}\cp_{\{s,t\}}  |f(s) - f(t)| - \frac{\alpha}{c}\sum_{\{s,t\} \in E_D}\dm_{\{s,t\}} |f(s) - f(t)|\right] \leq 0 \enspace.
	\]
    A solution is $c$-approximate if the value in the expectation above is non-positive, and such a solution can either be obtained by repeated rounding or by derandomization
    using the method of conditional expectations, without increasing the asymptotic run time.
\end{observation}

This observation implies that the bottleneck to obtaining a good approximation factor is the extent to which our rounding algorithm can approximate the marginal of $\mu_L$ on the individual edges of $D$.
Our main result relates this marginal to the combinatorial diameter of $\tset$. It can now be stated as follows:
\begin{theorem}
\label{thm:algo:main}
  Let $(G,D,\cp,\dm)$ be an instance of\ \,\scp, and $(\tset, \set{B_i}_i)$ a tree decomposition of $G$ with width $w(\tset)$ and combinatorial diameter $\Delta(\tset)$.
  Then {\sc SC-ROUND} satisfies $\algcut(s,t) \geq \Omega\paren{\frac{1}{\Delta(\tset)^2}} \cdot \lpcut(s,t)$ for every $\{s,t\}\in E_D$.  Therefore, we have a factor-$O(\Delta(\tset)^2)$ approximation for sparsest cut with run time $2^{O(w(\tset))}\cdot \poly(n)$.
\end{theorem}

The rest of this section is devoted to proving this theorem.
\subsection{Step 1: Reduction to Short Path}
In this section, we show that when the combinatorial diameter of the tree decomposition is $\cdiam = \Delta(\tset)$, the analysis can be reduced to the case of a path decomposition of length $\cdiam$.
We employ the following lemma to simplify our analysis of the behavior of the algorithm.
\begin{lemma}[\cite{ckr10}, Lemma 3.4]
  \label{lem:traverselInv}
	The distribution over the assignments $f$ is invariant under any connected traversal of $\tset$, i.e. the order in which bags are processed does not matter, as long as they have a previously processed neighbor.
	The choice of the first bag $B_0$ also does not impact the distribution.
\end{lemma}

Let $\set{s,t} \in E_D$ be a demand edge. %
If $s$ and $t$ are contained in a common bag, then $\algcut(s,t) = \lpcut(s,t)$ by \Cref{lem:bagRealisation} and we are done; therefore, we
assume that there is no bag containing both $s$ and $t$. %
We want to estimate the probability that $s$ and $t$ separated by the algorithm, that is, the probability that $f(s) \neq f(t)$.

The lemma above allows us to reduce to the case in which the algorithm first rounds a bag $B_1$ containing $s$, then rounds bags $B_2, \dots, B_{\ell-1}$ along the path to a bag $B_{\ell}$ containing $t$, and finally $B_{\ell}$.
At this point the algorithm has already assigned $f(s)$ and $f(t)$, so the remaining bags of $\tset$ can be rounded in any connected order without impacting the separation probability.
Hence, it is sufficient to characterize the behavior of the rounding algorithm along paths in $\tset$.

Let $P$ be the shortest path connecting a bag containing $s$ to a bag containing $t$; %
denote such path by $P = v_1 v_2 \ldots v_{\ell}$ such that $s \in B_{v_1}$ and $t \in B_{v_{\ell}}$. %
By \Cref{lem:traverselInv} we can assume that the algorithm first processes  $B_{v_1}$, and then all other bags $B_{v_2}, \dots, B_{v_\ell}$, in
this order.

Observe that, except for $v_1$ and $v_{\ell}$, no other bag of $P$ contains $s$ or $t$. 
We repeatedly apply the reduction rule from \Cref{def:simplification} until the resulting path has length at most $\cdiam$.
The following lemma asserts that the distribution of the algorithm is preserved under this reduction rule.

We slightly abuse the notation and denote by $\aset$ the distribution of our algorithm on path $P$ starting from $v_1$.

\begin{lemma}
\label{lem:distpreservereduction}
  Let $u,v,w$ be three consecutive internal bags on $P$ with $B_v \cap B_w \subseteq B_u$. Let $P'$ be a simplification of $P$ bypassing $v$ and let $\aset'$ be the distribution obtained by running the algorithm on path $P'$, starting on $v_1$. 
  Then $\aset'$ is exactly the same as $\aset$ restricted to $B(P')$.
\end{lemma}
\begin{proof}
  We can assume, without loss of generality, that $u,v,w$ appear on $P$ in the order of rounding; for otherwise, we apply \Cref{lem:traverselInv} twice: first, to reverse $P$, and preserve the distribution~$\aset$; then, to undo the reversing of $P'$ caused by the previous application.

  We modify the path decomposition $P$ into a (tree) decomposition $\hat{\tset}$ as follows: remove bag $v$ and add two new bags $v',v''$ where bag $v'$ is connected to $u$ and $w$ with $B_{v'} = B_u \cap B_v$ and $v''$  is connected to $v'$ with $B_{v''} = B_v$.
  This remains a tree decomposition for the vertices in $B(P)$ since vertices in $B_v \setminus B_u$ only occur in the bag $B_{v''}$ (due to our assumption that $B_v \cap B_w \subseteq B_u$).

  It is easy to check that run the algorithm {\sc SC-ROUND} on $\hat{\tset}$ produces exactly the same distribution as $\aset$.
  Since $B(P') = B(P) \setminus (B_{v''} \setminus B_{v'} )$, we have that $\aset|_{B(P')}$ is the distribution of {\sc SC-ROUND} on the path $\hat{P} = v_1 \ldots u v' w \ldots v_{\ell}$, obtained by removing $v''$ from $\hat{\tset}$.
  Now since $B_{v'} \subseteq B_u$, the rounding algorithm in fact does not do anything at bag $v'$, so it can be removed without affecting the distribution.
  We obtain path $P'$ as a result, and this implies that $\aset|_{B(P')}$ is the same distribution as $\aset'$.
\end{proof}

This result allows us conduct the rounding analysis on simplifications of paths.
It remains to show that this is beneficial, that is, that the rounding error can be bounded by the length of the path on which we round.
As in the work of \chlamtacetal~\cite{ckr10}, we use Markov flow graphs to analyze that error.

\subsection{Step 2: Markov Flow Graphs}
\label{sec:markov}
Let $P = v_1, \dots, v_\ell$ be a path with length $\ell$ and $s \in B_{v_1}$, $t\in B_{v_\ell}$. %
We run \autoref{alg:Chlamtac} from $v_1$ to~$v_\ell$ to compute some assignment $f$.
Let $\aset$ be the probability distribution of the resulting assignment $f$. %
Recall that $\algcut(s,t)$ denotes the probability that the algorithm assigns $f(s) \neq f(t)$, and $\lpcut(s,t)$ is the probability that $s$ and $t$ are separated according to the distributions $\set{\mu_L}_L$, i.e.~$\Pr_{f \sim \mu_{B \cup \{s,t\}}}[f(s) \neq f(t)]$.
In the second step, we analyze the probability of $\algcut(s,t)$ in terms of $\lpcut(s,t)$.
This step is encapsulated in the following lemma.

\begin{lemma}
\label{lem:algo:markov}
  There exists a directed layered graph $H$ containing nodes $s_0,s_1,t_0,t_1 \in V(H)$ and a weight function $w_H$ on the edges, satisfying the following properties:
  \begin{enumerate}
    \item For $i = 0,1$, we have that $\Pr_{f \sim \aset}[f(s) = i\; \&\; f(t) = 1-i]$ is at least an $\Omega(1/\ell^2)$-fraction of the minimum $(s_i, t_{1-i})$-cut of $H$.\label{lem:algo:markov:i}
    \item For $i = 0,1$, the value of a maximum $(s_i, t_{1-i})$-flow in $H$ is at least $\Pr_{f \sim \mu}[f(s) =i\; \&\; f(t) = 1-i]$.\label{lem:algo:markov:ii}
  \end{enumerate}
\end{lemma}

\Cref{thm:algo:main} immediately follows from this lemma.
\begin{proof}[Proof of \Cref{thm:algo:main}]
    We run the algorithm of \chlamtacetal to get some $V(G)$-assignment~$f$.

    Consider a pair $\set{s,t} \in E_D$.
    Using \Cref{lem:traverselInv} and \Cref{lem:distpreservereduction}, we can reduce the analysis to a path $P$ of length at most $\cdiam$, which is a simplification of a path in $\tset$. %
    Now, by \Cref{lem:algo:markov} and max-flow-min-cut theorem, we get that 
    \begin{align*}
    \algcut(s,t)
    &= \Pr_{f \sim \aset}[f(s) = 0\; \&\; f(t) = 1] + \Pr_{f \sim \aset}[f(s) = 1\; \&\; f(t) = 0] \\
    &\geq \Omega\paren{\frac{1}{\cdiam^2}} \paren{\operatorname{mincut}(s_0, t_1) + \operatorname{mincut}(s_1, t_0)} \\
    &= \Omega\paren{\frac{1}{\cdiam^2}} \paren{\operatorname{maxflow}(s_0, t_1) + \operatorname{maxflow}(s_1, t_0)} \\
    &\geq \Omega\paren{\frac{1}{\cdiam^2}} \paren{\Pr_{f \sim \mu}[f(s) =0\; \&\; f(t) = 1] + \Pr_{f \sim \mu}[f(s) =1\; \&\; f(t) = 0]} \\
    &= \Omega\paren{\frac{1}{\cdiam^2}} \lpcut(s,t) \enspace .
    \end{align*}
    Therefore, $f$ separates each pair $\set{s,t}$ with probability that is a factor of $O(\cdiam^2)$ away from $\lpcut(s,t)$.

    Applying Observation \ref{obs:derandomisation} with $c = \Omega(1/\cdiam^2)$, we can obtain (deterministically) an assignment $f^*$ that is an $O(\cdiam^2)$-approximation for the \scp instance.
\end{proof}

The rest of this section is dedicated to proving the \Cref{lem:algo:markov}.
The tools needed for this proof are implicit in the work of \chlamtacetal~\cite{ckr10}.
We restate them for the sake of completeness and in order to adjust it to our terminology. 

The section is organized as follows: first, we describe the construction of our graph $H$, and then we proceed to analyze the values of maximum flow and minimum cut.
We will only analyze the flow and cut for $i= 0$, that is, $(s_0,t_1)$-flow and $(s_0,t_1)$-cut. The other case is analogous.

\paragraph{Construction of Graph $H$:}
Without loss of generality, we can assume that the distributions $\set{\mu_L}_L$ are symmetric in the labels $\{0,1\}$, see \Cref{lem:symmetrization}.
In particular, this gives $\Pr[f(v) = 1] =\linebreak \Pr[f(v) = 0] = 1/2$ for any vertex $v$.

The rounding can be modeled by a simple Markov process. %
Denote by $I_0,\dots, I_\ell$ the sets that are conditioned on in \autoref{alg:Chlamtac}, $I_i = B_{v_i} \cap B_{v_{i+1}}$ for $i \in \set{1,\ldots,
\ell-1}$; we refer to these sets as \emph{conditioning sets}.

For the initial and final sets of the rounding procedure we take $I_0 = \{s\}$, $I_\ell = \{t\}$.
Now we are ready to describe our graph $H$:
\begin{itemize}
    \item \textbf{Vertices}: %
    Vertices of $H$ are arranged into layers $L_0,\dots,L_\ell$ with $L_i = \mathcal{F}[I_i]$. 
    Observe that $|L_i| = 2^{|I_i|}$. %
    The vertices of $H$ represent the intermediate states the algorithm might
    reach.

    \item \textbf{Edges}: %
    For each $i$, there is a directed edge from every vertex in $L_i$ to every vertex in $L_{i+1}$. %
    The weight of the edge $(f_i, f_{i+1})$, for $f_i \in L_i$, $f_{i+1} \in L_{i+1}$, is equal to the  probability of joint event,
    $w_H(f_i, f_{i+1}) = \Pr[f|_{I_i} = f_i \wedge  f|_{I_{i+1}} = f_{i+1}]$.

    We remark that the weight is 0 whenever $f_i$ and $f_{i+1}$ are
    contradictory, and that probabilities are well defined, as $I_i \cup I_{i+1} \subseteq B_{i+1}$.
\end{itemize}

Observe that the weight of an edge is the probability that both of its endpoints are reached by the algorithm, and hence the probability that the algorithm transitions along that edge.

\newcommand{\iset}{{\mathcal I}}

\begin{observation}
  Let $\iset = \bigcup_{i} I_i$.
  The distribution $\aset|_{\iset}$ can be viewed as the following random walk in~$H$: Pick a random vertex in $L_0$ and start taking a random walk where each edge is taken with probability proportional to its weight. %
  Formally, once a node $f_i$ is reached, choose the next node $f_{i+1}$ with probability $w_H(f_i, f_{i+1})/\Pr[f|_{I_i} = f_i]$.
\end{observation}

At this point, we rename $\aset := \aset|_{\iset}$.
Notice that the layer $L_0$ contains two vertices corresponding to the assignment $f(s) = 0$ and $f(s) = 1$, respectively. We denote them by $L_0 = \{s_0, s_1\}$. %
Similarly, $L_{\ell} = \{t_0, t_1\}$. %
Notice further that $\Pr_{f \sim \aset}[f(s) = 0, f(t) = 1]$ is exactly the probability that the random walk starts at $s_0 \in L_0$ and ends at $t_1 \in L_{\ell}$.

\paragraph{Maximum $(s_0,t_1)$-Flow:}
We are now ready to show that the value of the maximum $(s_0,t_1)$-flow is at least $\Pr_{f \sim \mu}[f(s) =0, f(t)=1]$.

We define the flow $g\colon E(H) \to \mathbb{R}_{\geq 0}$ as follows, for $i \in \set{1,\ldots,
\ell-1}$, $f_i \in L_i$ and $f_{i+1} \in L_{i+1}$:
\[
	g(f_i, f_{i+1}) = \Pr_{f\sim \mu_{B_{v_{i+1}} \cup \{s,t\} }}[ f(s) = 0, f(t) = 1, f|_{I_i} = f_i, f|_{I_{i+1}} = f_{i+1} ] \enspace .
\]
We remark that $g$ is an $s_0$-$t_1$-flow, that is, it satisfies flow conservation at all vertices in $H$ except $s_0,t_1$, and the capacities of
graph $H$ are respected, that is, $g(e) \leq w_H(e)$ for all $e \in E(H)$. %
The value of~$g$ is given by: %
\begin{align*}
	\sum_{(s_0,f^*) \in \delta^+(s_0)}g(s_0, f^*) &= \sum_{f^* \in \mathcal{F}[I_1]}\Pr_{f  \sim \mu_{B_{v_1} \cup \{s,t\} } }[f(s) = 0, f(t) = 1, f|_{I_1} = f^*]\\
	&=\Pr_{f  \sim \mu_{B_{v_1} \cup \{s,t\} } }[f(s) = 0, f(t) = 1] \enspace .
\end{align*}

This concludes the proof of Point \ref{lem:algo:markov:ii} of \Cref{lem:algo:markov}.

\paragraph{A Potential Function:}
Before we show a cut with the desired capacity, we need to introduce some notation.
For $i = 0,\ldots, \ell$, let $X_i$ be a random variable indicating the vertex in $L_i$ visited by the random walk (i.e.~picked by the algorithm.
We denote by $\textbf{X}= X_0 X_1 \ldots X_{\ell}$ the path taken in the random walk process. %
We can interchangeably view distribution $\aset$ as either the distribution that samples an assignment $f\colon \iset \rightarrow \{0,1\}$ or one that samples a (random walk) path $\textbf{X}$.

We define, for every layer $L_i$ and every vertex $v \in L_i$,
\[
	A(v):= \Pr_{\textbf{X} \sim \aset} [X_0 = s_0 \mid X_i = v] -\frac{1}{2} \enspace .
\]
Intuitively, this function captures the extent to which $v$ has information about the initial state of the Markov process. %
On the one hand, if $A(v)$ is equal to $0$, $v$ knows essentially nothing about $X_0$, the choice of $v$ does not imply anything about $X_0$. %
On the other hand, if $A(v)$ is far from $0$, then we can glean a lot of information about $X_0$ from $v$ being visited; %
in particular, if the probability that $s$ and $t$ are cut is low, we must have $A(t_1) \approx - 1/2$.

To track how $A$ changes from layer to layer, we use the potential function $\phi\colon \{0,\dots,\ell\} \to \mathbb{R}_{\geq 0}$, defined as:
\[
	\phi(i) := \operatorname{Var}_{\textbf{X} \sim \aset}[A(X_i)] \enspace. 
\]

The following lemma by \chlamtacetal bounds the change in potential in terms of the probability that $X_0 = s_0$ and $X_\ell = t_1$.
\begin{lemma}[\cite{ckr10}, Lemma 5.2]
	It holds $\phi(0) - \phi(\ell) \leq 2\Pr[X_0 = s_0 \wedge X_\ell = t_1]$ \enspace .
\end{lemma}

\paragraph{Minimum $(s_0,t_1)$-cut:}
We are now ready to analyze the value of minimum $(s_0,t_1)$-cut in $H$. It suffices to give a lower bound on $\phi(0) - \phi(\ell)$.
This is is possible by the following lemma which is proved implicitly by \chlamtacetal~\cite{ckr10}.
\begin{lemma}[\cite{ckr10}, Lemma 5.4]
	\label{lem:MarkovCut}
	Let $C$ be the set of edges $(f_i, f_{i+1})$ in $E(H)$ such that $|A(f_i) - A(f_{i+1})|$ is at least some threshold $\rho > 0$. Then $\sum_{e \in C}w_H(e) \leq (\phi(0) - \phi(\ell)) \cdot 1/\rho^2$.
\end{lemma}

We can apply \Cref{lem:MarkovCut} in the following fashion.
Suppose $A(t_1) \geq 0$. In that case we have $\Pr[X_0 = s_0 \mid X_\ell = t_1] \geq 1/2$, so $s$ and $t$ are cut with probability at least $\frac{1}{2}\lpcut(s,t)$.
This error is already a small enough, so assume $A(t_1) < 0$. Then $A(s_0) - A(t_1) > 1/2$.
Since every path from $s_0$ to $t_1$ has exactly $\ell$ edges, any such path must contain an edge $(f_i, f_j)$ with $A(f_i) -A(f_j) > 1/ (2\ell)$.
Cutting all such edges therefore separates $s_0$ and $t_1$. Hence, by applying \Cref{lem:MarkovCut}, the minimum $s_0$-$t_1$-cut has size at most 
\begin{align*}
	O(\ell^2) (\phi(0) - \phi(\ell))
	&\leq O(\ell^2) \Pr[X_0 = s_0 \wedge X_\ell = t_1]\\
	&=O(\ell^2) \Pr[f(s) = 0 \wedge f(t) = 1] \enspace .
\end{align*}

This concludes the proof of Point \ref{lem:algo:markov:i} of \Cref{lem:algo:markov}. We see that the cutting probability predicted by the distributions is realised by the rounded solution $f$, up to a factor $\Omega(1/\ell^2)$.

This gives an alternative to the analysis given by \chlamtacetal whose constant depends on the size of the layers of $H$ rather than the number of layers.
While the layer sizes depend only on $k$, the dependence is exponential.
The number of layers is a priori $\log(n)$, which would give a worse approximation guarantee.
However, we will show how to modify a tree decomposition to ensure that $H$ has few layers.

\section{Combinatorially Shallow Tree Decompositions}

In this section, we show how to construct tree decompositions with low combinatorial diameter, thus achieving the approximation results stated in \Cref{thm:main-intro}. %
We start by restricting our consideration to decompositions that are shallow in the traditional sense. %
For a given graph $G$ with treewidth $k$, we consider a tree decomposition $(\tset, \set{B_i}_{i \in V(\tset)})$ with diameter $\diam=O(\log n)$ and width $O(k)$~\cite{Bodlaender88}. %
Fix some root $r$ in $V(\tset)$.

Our goal is now to modify $\tset$ such that every node has a combinatorially short path to $r$. %
This is a necessary requirement, but perhaps surprisingly it is not sufficient. %
The combinatorial lengths of paths do not necessarily induce a metric on $V(\tset)$%
\footnote{Consider bags $\{ab\}, \{abc\}, \{acd\}, \{ade\}, \{aef\}, \{afg\}, \{a\}$ occuring in that order as a path. The whole path can be reduced to just the endpoints. The subpath $\{ab\}, \{abc\}, \{acd\}, \{ade\}, \{aef\}, \{afg\}$ is irreducible. Thus the distance from $\{ab\}$ to $\{afg\}$ is larger than the sum of the distances from $\{ab\}$ to $\{a\}$ and $\{afg\}$ to $\{a\}$.}, %
and therefore bounding the length to $r$ does not on its own suffice to bound the combinatorial diameter.

We will not show explicitly that the modified structures are in fact tree decompositions.
The proofs are straightforward using \Cref{lem:TreeDecompMonotonicity}. 

We introduce three objects, which we call \textbf{bridges}, \textbf{highways}, and \textbf{super-highways}, and show that they can be used to prove the three parts of \Cref{thm:main-intro}. 

\subsection{Bridges}
\begin{figure}
	\includegraphics[width=\textwidth]{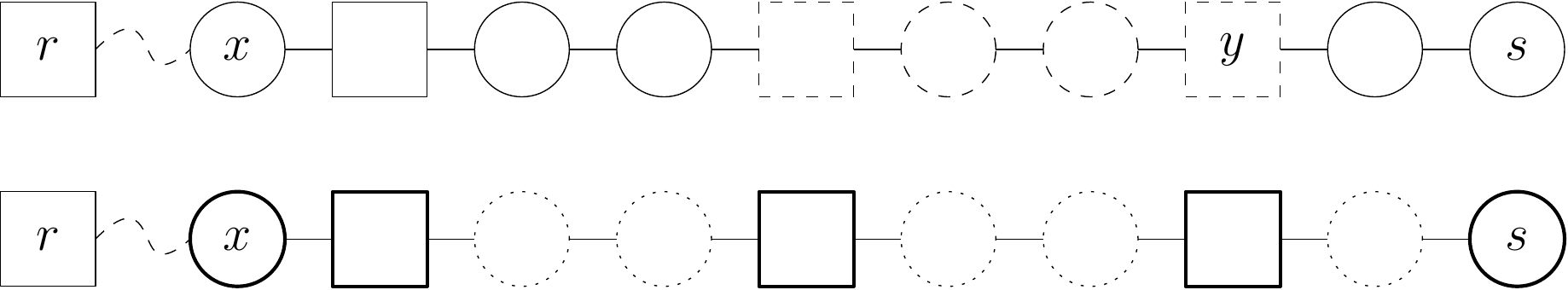}
	\caption{Illustration of a path from the root to some node $s$. The square nodes are the synchronization nodes. The bridge from $y$ to its synchronization ancestor is marked with dashes in the first image. The dotted nodes in the second image mark those nodes which can be removed when simplifying the $x$-$s$-path in $\tset'$.  }
	\label{fig:bridges}
\end{figure}
Fix a parameter $\lambda \in \set{1,\dots, \diam}$. Define $\ell\colon V(\tset) \to \mathbb{N}_0$ to be the \emph{level} of a node in $\tset$, that is, $\ell(v)$ is the number of edges on $\tset_{v \leftrightarrow r}$.

\begin{definition}
	We call a node a \emph{synchronization node} if its level is a multiple of $\lambda$.
	Define also the \emph{synchronization ancestor} $\sigma(v)$ of any node $v$ to be the first node on the path from $v$ to $r$ that is a synchronization node, excluding $v$ itself.
\end{definition}

We can construct a tree decomposition $(\tset', \set{B'_i}_i)$ by taking $\tset'=\tset$ and setting $B_v' = B(\tset_{v \leftrightarrow \sigma(v)})$, that is, the new bag is obtained by combining all the bags from $v$ up to its synchronization ancestor. %
This increases the width of the decomposition by a factor of at most $\lambda$. We may view this path connecting $v$ to the synchronization point as a {\bf bridge} crossing over all intermediate nodes in one step.

\begin{lemma}
	\label{lem:bridges:diam}
	$\tset'$ has combinatorial diameter $O(\diam / \lambda)$.
\end{lemma}
\begin{proof}
	Fix any two nodes $s,t \in V(\tset')$ and take $x$ to be their lowest common ancestor in $\tset'$.
	Then the combinatorial length of $\tset'_{s \leftrightarrow t}$ is at most the sum of the combinatorial lengths of $\tset'_{s \leftrightarrow x}$ and $\tset'_{x \leftrightarrow t}$.%
	We remark that triangle inequality holds in this case, because $x$ is on the path from $s$ to $t$. %
	Thus, it suffices to show that the combinatorial length of $\tset'_{s \leftrightarrow x}$ is $O(\diam / \lambda)$.
  The result follows analogously for $\tset'_{x \leftrightarrow t}$.

  Using the rules of \Cref{def:simplification}, we can bypass any node that is neither a synchronization node nor $s$ or $x$, since the bag of the unique child (in $\tset'_{s \leftrightarrow x}$) of such a node is a superset of its own bag. %
  Therefore, the path $\{v \in \tset'_{s \leftrightarrow x} | v = s \vee v=x \vee v\text{ is a synchronization node}\}$ is a simplification of~$\tset'_{s \leftrightarrow x}$.
  Since there are at most $\diam/\lambda$ synchronization nodes on any upward path, the lemma follows.
\end{proof}

This lemma, in conjunction with \Cref{thm:algo:main} and the fact that $\tset'$ can be computed in polynomial time from $\tset$, yields:
\begin{corollary}
	\label{cor:bridge:algo}
	For every $\lambda$, there is an algorithm that computes an $O((\log n / \lambda)^2)$-approximation for \scp instances where $G$ has treewidth at most $k$, in time $2^{O(\lambda k)}\poly(n)$.

	Setting $\lambda = \log n /k$ results in an $O(k^2)$-approximation in time $2^{k}\poly(n)$, while setting $\lambda=\log n$ gives an $O(1)$-approximation in time $n^{O(k)}$. 
\end{corollary}

\subsection{Highways}
\begin{figure}
	\includegraphics[width=\textwidth]{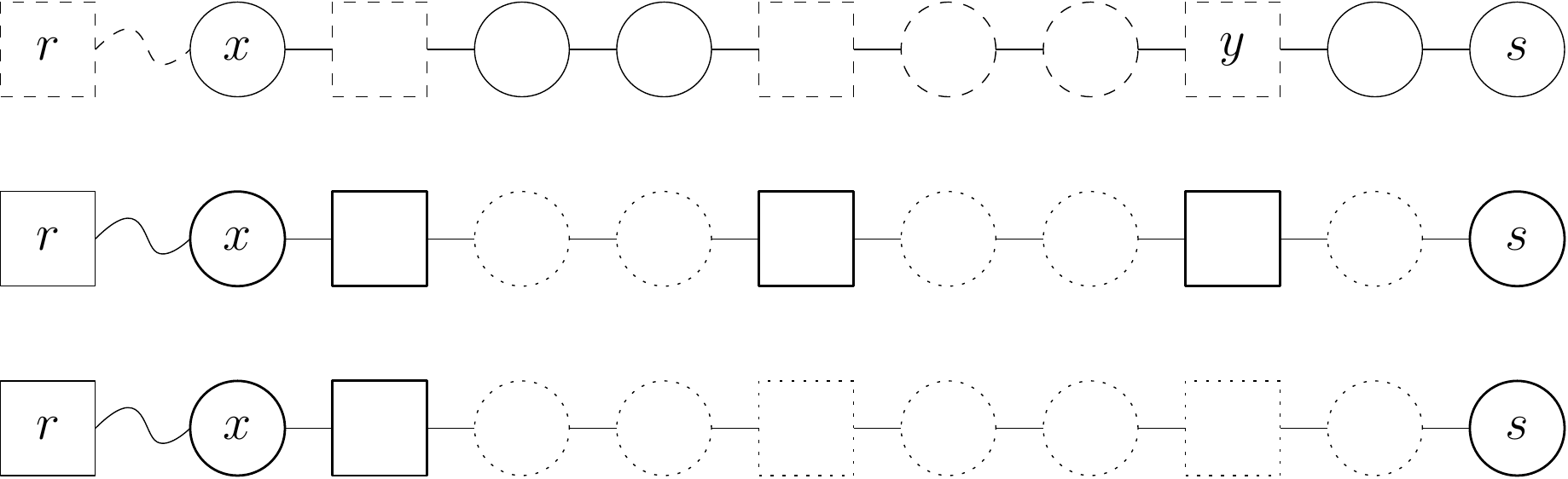}
	\caption{The dashed nodes in the first image mark the bridge and highway from $y$ to $r$. The other images illustrate the two simplification rounds for the $x$-$s$-path, leaving a path of length $2$.}
	\label{fig:highways}
\end{figure}

The idea of extending bags towards the root can be exploited further by adding the vertices in a synchronization bag to all of its descendants. 
We may regard this as giving each node a bridge to the next synchronization node, as well as a {\bf highway} along the synchronization nodes towards the root.
This idea leads to the following construction.

Let $(\tset', \set{B'_i}_i)$ be a modified tree decomposition with $\tset'=\tset$ as before, and
\[
B_v' := B(\{ w \in \tset_{v \leftrightarrow r} \mid w \in \tset_{v \leftrightarrow \sigma(v)} \vee w \text{ is a synchronization node} \}) \enspace .
\]

The size of these bags is at most $k (\lambda + \diam / \lambda)$, which for $\lambda = \diam /k$ gives $\diam+ k^2 = O(\log n + k^2)$.

Notice that the bag $B_r$ is now contained in any bag $B'_i$, so we have some hope that the combinatorial diameter of $(\tset', \set{B'_i}_i)$ is low.
Indeed this is true.

\begin{lemma}
	\label{lem:highways:diam}
	$\tset'$ has combinatorial diameter at most $3$.
\end{lemma}
\begin{proof}
	As before, we split any $s$-$t$-path at $x$, the lowest common ancestor of $s$ and $t$, and consider only the $s$-$x$-path.
	Every non-synchronization node $v$ on $\tset_{s\leftrightarrow x}$ has a node below it which is either a synchronization node or $s$.
	The bag of that node is a superset of $B'_v$, so all non-synchronization nodes except $s$ and $x$ can be bypassed.
	Call that reduced path $P$.

  Consider the neighbor of $s$ in $P$, which we denote $v$, and assume that $v$ is not the neighbor of $x$ in $P$. %
  Then $v$ must be a synchronization node, and its next node in $P$ is $\sigma(v)$. %
  Now, the intersection $B'_v \cap B'_{\sigma(v)}$ contains exactly all of the bags of synchronization nodes in    $\tset_{\sigma(v)\leftrightarrow r}$, and thus, ${B'_v \cap B'_{\sigma(v)} \subseteq B'_s}$. %
  This implies that $v$ can be bypassed, and by repeating this process, we can bypass every synchronization node except for the neighbor of $x$.

	This gives a possible simplification of $\tset_{s \leftrightarrow t}$ as the path $(s, \sigma_s, x, \sigma_t, t)$, where the $\sigma_s$ and $\sigma_t$ are the synchronization nodes below $x$ on the paths to $s$ and $t$, respectively.
	There is a further reduction of the whole path, since $B_x'$ is precisely $B_{\sigma_s}' \cap B_{\sigma_t}'$.
	This allows us to remove $x$ as well, giving a simplification of length $3$.
\end{proof}

Using the fact that $\diam \in O(\log n)$, and setting $\lambda = \diam /k$ gives a fixed-parameter algorithm that yields a constant-factor approximation:
\begin{corollary}
	\label{cor:highways:algo}
	There exists an algorithm that in time $2^{O(k^2)}\cdot\poly(n)$ computes a factor-$O(1)$ approximation for \scp instances where $G$ has treewidth at most $k$.
\end{corollary}

\subsection{Super-Highways}
\begin{figure}
	\includegraphics[width=\textwidth]{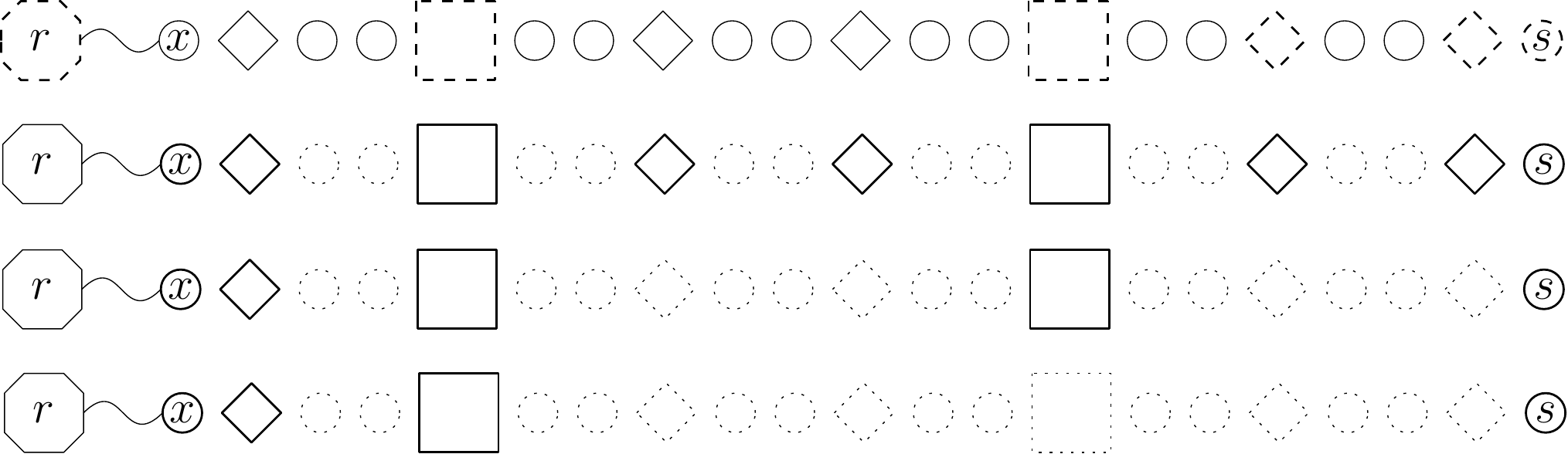}
	\caption{Illustration of an upward path with nodes of layer $-1$ as circles, nodes of layer 0 as diamonds, and nodes of layer $1$ as squares. The root is at some unspecified maximum layer. The dashed nodes in the first image mark the super-highway from $s$ to $r$. The other images illustrate the simplification rounds for the $x$-$s$-path, removing all nodes of some layer in each round, except $x$, $s$, and possibly one node close to $x$.}
	\label{fig:superhighways}
\end{figure}

We can think of the previous construction as having two layers, bridges to synchronization nodes and highways along synchronization nodes to the root. The highways need to cover many synchronization nodes, leading to large bags in $\tset'$.
To improve on this we introduce a network of {\bf super-highways} of different layers, where each layer covers fewer, more spaced-out synchronization nodes on a root-leaf path.
When we connect a node to the root we can then move up the tree layer by layer with increasing speed, decreasing the size of bags in $\tset'$. 
This is payed for by the need for an additional node in path simplifications for moving between layers, giving a trade-off between run time and approximation guarantee.

Let $q \in \mathbb{N}$ be a parameter representing the number of layers. %
For a node $v \in \tset$, we define the layer of $v$, denoted $\pi(v)$, as
\[
	\pi(v) := \max\{-1, \max\{ j\in\{0,\dots, q-1 \} \mid \ell(v) \equiv 0 \mod k^{j/q} \diam / k\} \} \enspace .
\]
By this definition all synchronization nodes are assigned to some non-negative layer, and all other nodes are on layer $-1$.
We now get a new tree decomposition $(\tset', \set{B'_i}_i)$ by constructing bags:
\[
B_v' = B(\{ w \in \tset_{p(v) \leftrightarrow r} \mid \pi(w) = \max\{ \pi(u) \mid u \in \tset_{p(v)\leftrightarrow w}  \} \} \cup \{v\})\enspace .
\]
Informally, we start at some node $v$ and move towards $r$ by first taking all nodes of layer $-1$ until we hit a node of layer $0$, then taking only nodes of layer $0$ until we hit layer $1$, and so on.
The nodes at higher layers are spaced further apart. Thus this process ``speeds up'' thereby generating smaller bags.
To be precise, there are  $q$ layers and at most $k^{1/q}$ nodes of any one layer in a bag, so $\tset'$ has width $O(\diam + qk^{1+1/q})$.

We now show that $(\tset', \set{B'_i}_i)$ has combinatorial diameter depending only on $q$.

\begin{lemma}
\label{lem:superhighways:diam}
	$(\tset', \set{B'_i}_i)$ has combinatorial diameter at most $2q+1$.
\end{lemma}
\begin{proof}
	As before, we only show that any upward path from $s$ to $x$ has combinatorial length at most~$q+1$.
	We need to perform a round of reductions for every layer, with the goal of leaving only $s$, $x$, as well as the first node of at least that layer below $x$.
	For layer $-1$, this holds with the same argument as before.

	We can now proceed by induction, fixing some layer $i$ and assuming that the $s$-$x$-path $P$ has been reduced to contain only $s$, then nodes of layers $\geq i$, followed by a sequence $(\sigma_{i-1}, \sigma_{i-2}, \ldots, \sigma_0, x)$, where each node $\sigma_j$ is in layer $j$. Here, we assume w.l.o.g.\ that $x$ is at layer $-1$.
	Now consider any node $v$ of layer $i$, except the one closest to $x$.
	Because its neighbors also have level at least $i$ (or are~$s$), the intersection of their bags can be represented as the union of bags of $\tset$ whose layer is at least $i$.
  Let~$w$ be the predecessor of $v$ on $s$-$x$-path $P$.
	The set $B_w'$ is constructed from some upward path starting at $w$, containing only nodes of non-decreasing layer.
	This upward path hits layer~$i$ between~$w$ and~$v$, but not layer $i+1$ since a node of layer $i+1$ would be on $P$ between $w$ and~$v$.
	So then~$B_w'$ covers all nodes of layer at least $i$ that $B_v'$ covers, and therefore $v$ can be bypassed, concluding induction.

	The simplification of $\tset_{s\leftrightarrow x}$ produced in this fashion is a path $(s, \sigma_{q-1}, \dots, \sigma_0, x)$, where $\pi(\sigma_i) = i$.

	If we add the same simplification for $\tset_{t\leftrightarrow x}$ we get a simplification for $\tset_{s\leftrightarrow t}$ that takes the form $(s, \sigma_{q-1}, \dots, \sigma_0, x, \sigma_0',\dots,\sigma_{q-1}',t)$.
	As before $x$ can be bypassed since its bag is the intersection of the bags of $\sigma_0$ and $\sigma_0'$.
	Thus any $s$-$t$-path in $\tset'$ has combinatorial length at most $2q+1$.
\end{proof}

This implies the existence of the following algorithms.
\begin{corollary}
\label{lem:superhighways:algo}
	There exists an algorithm that, for any $q\in \mathbb{N}$, computes a factor-$O(q^2)$ approximation for \scp in time $O(2^{qk^{1+1/q}})\cdot\poly(n)$.
	Taking $q = \log k$ gives a factor-$O(\log^2k)$ approximation in time $2^{O(k\log k)}\cdot\poly(n)$.
\end{corollary}

\begin{acks}
Parinya Chalermsook has been supported by European Research Council (ERC) under the
European Union’s Horizon 2020 research and innovation programme (grant
agreement No. 759557) and by Academy of Finland Research Fellowship, under
grant number 310415. %
Joachim Spoerhase has been partially supported by European Research Council (ERC) under the European Union’s Horizon 2020 research and innovation programme (grant agreement No. 759557). %
Daniel Vaz has been supported by the Alexander von Humboldt Foundation with
funds from the German Federal Ministry of Education and Research (BMBF).”
\end{acks}

\bibliographystyle{abbrv}
\bibliography{references}

\appendix
\section{Various Lemmas}

\makeatletter
\newcommand{\oset}[3][0ex]{%
  \mathbin{\mathop{#3}\limits^{
    \vbox to#1{\kern-2\ex@
    \hbox{$\scriptstyle#2$}\vss}}}}
\makeatother
\newcommand{\mirror}[1]{\ensuremath{\oset{\leftrightarrow}{#1}}}

\def\vecsign{\raise-0.867ex\hbox{$\mathchar"017E$}}
\newcommand{\dvec}{\normalfont\mathrel{\ooalign{\vecsign\cr\hidewidth\raise1ex\hbox{\rotatebox{180}{\vecsign}}\cr}}}
\renewcommand{\mirror}[1]{\ensuremath{\oset[-0.2ex]{\dvec}{#1}}}

\begin{definition}
	For any set $X$ and $X$-assignment $f$ we define the \emph{mirror} of $f$ to be $\mirror{f} := \sigma \circ f$, where $\sigma(0) = 1$ and $\sigma(1) = 0$.
\end{definition}

Notice that the mirror of an assignment represents the same cut; it merely exchanges the labels.
The approximation ratio analysis of \chlamtacetal requires the distributions to be symmetric in their labeling, in particular $\Pr[f(v) = 0] = \Pr[f(v) = 1]$ $\forall v$.
They resolve this by demanding symmetry via the Sherali-Adams LP which can be shown to not worsen the relaxation.
Using the following lemma, we are able to prove that the rounding analysis also holds in the non-symmetric case.

\begin{lemma}
	\label{lem:symmetrization}
	Let $G, (\tset, \set{B_i}_{i}), \{\mu_L\}$ be the input of \autoref{alg:Chlamtac} and $f$ the assignment computed by it.
	Consider the modified decomposition $(\tset, \set{B'_i}_i)$, where a dummy vertex $e$ has been added to every bag $B'_i$.

	Then for each $ \mu_L$ and $L' = L \cup \{e\}$ there exists a $\mirror{\mu}_{L'}$ with $\Pr_{ f \sim \mirror{\mu}_{L'} }[ f = f' ] = \Pr_{ f \sim \mirror{\mu}_{L'} }[ f = \mirror{f'} ]$ for all $f' \in \mathcal{F}[L']$ such that when \autoref{alg:Chlamtac} is run on $G, (\tset, \set{B'_i}_i), \{\mirror{\mu}_{L'}\}$ the resulting assignment $f^*$ satisfies
	\[
		\Pr[f = f'] + \Pr[f = \mirror{f'}] = \Pr[f^*|_{V(G)} = f'] + \Pr[f^*|_{V(G)} = \mirror{f'}] \; \forall f' \in \mathcal{F}[V(G)] \enspace.
	\]
\end{lemma}

The content of the Lemma is at its core not very surprising.
If we do not care about the labels, we do not care about whether the algorithm outputs $f$ or $\mirror{f}$.
But if that is the case, the distributions also should not need to maintain some distinction between the labels.
In fact, one could run the algorithm unmodified, and then permute the labels with probability $1/2$.
Clearly this does not change the distribution over cuts, and the choice of labels is now symmetric.

\begin{proof}
	To make this formal we shall use the value of $f(e)$ to indicate whether or not we are permuting the labels.
	Consider the following definition:
	\[
	\mirror{\mu}_{L\cup \{e\}} (f', e \to 0) = \frac{1}{2}\mu_L(f'), \; \mirror{\mu}_{L\cup \{e\}} (f', e \to 1) = \frac{1}{2}\mu_L(\mirror{f'}) \; \forall f' \in \mathcal{F}[L] \enspace.
	\]
	This definition describes a distribution with the desired symmetry property.
	By \Cref{lem:traverselInv} we can model the rounding algorithm for $G, (\tset, \set{B'_i}_{i}), \{\mirror{\mu}_{L'}\}$ as choosing first a value for $f(e)$, and then proceeding in the same order as the rounding over $G, (\tset, \set{B_i}_{i}), \mu_{L}\}$.
	With probability $1/2$ we get ${f(e) = 0}$. Since every bag contains $e$, $e$ is always conditioned on, so the symmetrized algorithm performs the exactly as the original run would.
	Meanwhile if $f(e) = 1$, the symmetrized algorithm samples some intermediate assignment $f'$ with exactly the probability that the original algorithm would have sampled $\mirror{f'}$.
	This gives
	\begin{align*}
		&\Pr[f^*|_{V(G)} = f'] = \frac{1}{2}\Pr[f = f'] + \frac{1}{2}\Pr[\mirror{f} = f'] \; &\forall f' \in \mathcal{F}[V(G)] \\
		\implies &\Pr[f^*|_{V(G)} = f'] + \Pr[f^*|_{V(G)} = \mirror{f'}] = \Pr[f = f'] + \Pr[f = \mirror{f'}] \; &\forall f' \in \mathcal{F}[V(G)].
	\end{align*}
\end{proof}

Notice that while we could construct the symmetrized $\mirror{\mu}$ efficiently, we do not need to.
The mere existence of a symmetrized set of distributions is sufficient for our purposes.
The analysis of the Markov flow graphs in \Cref{sec:markov} requires symmetry, but by the lemma above we can assume symmetry without loss of generality.
The result then also holds for the non-symmetric case since the probability that an edge is cut is symmetric in the labels by
\[
\Pr[f(s) \not = f(t)] = \Pr[f(s) = 1 \wedge f(t) = 0] + \Pr[f(s) = 0 \wedge f(t) = 1] \enspace .
\]

\begin{lemma}
	\label{lem:TreeDecompMonotonicity}
	Let $(\tset,\{B_i\}_{i\in V(G)})$ be a tree decomposition of a graph $G$, rooted at $r$.
	Then ${(\tset, \{B_i'\}_{i\in V(G)})}$ is also a tree decomposition of $G$ if $B_i \subseteq B_i' \subseteq B_i \cup B_{p(i)}'$.
\end{lemma}
\begin{proof}
	Fix some $s \in V(G)$. We need to show that $\tset'_s := \{i \in V(\tset) \mid s \in B_i'\}$ is connected.
	As $\tset_s := \{i \in V(\tset) | s \in B_i\}$ is connected and $\tset_s \subseteq \tset'_s$, it is sufficent to show that any $i\in \tset'_s$ is connected to $\tset_s$ in $\tset'_s$.
	We do this by induction over the distance of $i$ to the root.

	For $i = r$ we have $B'_r = B_r$, so either $i \not \in \tset'_s$ or $i \in \tset_s$.
	Otherwise, consider some $i \in \tset'_s$, so $s \in B_i \cup B'_{p(i)}$.
	Then we either have $s \in B_i$, in which case we are done, or $s \in B_{p(i)}'$.
	But this gives $p(i) \in \tset'_s$, and $p(i)$ is closer to $r$ than $i$.
	Thus we can assume that $p(i)$ is connected to $\tset_s$, and hence~$i$ is also connected to $\tset_s$ via $p(i)$.
\end{proof}

\end{document}